\documentclass[aps,pra,letterpaper,twocolumn,showpacs,superscriptaddress,floatfix,longbibliography]{revtex4-1}
\usepackage[english]{babel} \usepackage{amsmath} \usepackage{color}
\usepackage{epsfig} \usepackage{graphicx} \usepackage{bm}
\usepackage{mathtools}

\usepackage{times,amsmath,amssymb} \usepackage{amsfonts}
\usepackage{amssymb}

\newcommand{\tr}{\mathop{\mathrm{tr}}\limits}
\newcommand{\traccazero}{\mathop{\mathrm{tr}_{\mathcal{H}_0}}\limits}

\newcommand{\eket}[2]{|e^#1 #2 \rangle}
\newcommand{\ebra}[2]{\langle e^#1 #2|}
\newcommand{\hc}{\mathop{\mathrm{H.c.}}\limits}

\newtheorem{theorem}{Theorem} \newtheorem{lemma}{lemma}
\newenvironment{proof}[1][Proof]{\noindent\textit{#1.} }{\
  \rule{0.5em}{0.5em}}

\global\long\def\bra#1{\langle{#1}|}
\global\long\def\ket#1{|{#1}\rangle }
\global\long\def\braket#1#2{\langle#1|#2\rangle}
\global\long\def\ketbra#1#2{\ket{#1}\bra{#2}}

\global\long\def\al{\alpha} \global\long\def\be{\beta}
\global\long\def\ga{\gamma}

\global\long\def\la{\lambda} \global\long\def\ka{\kappa}
 \global\long\def\Om{\Omega}
\global\long\def\eps{\epsilon}

\begin{document}

\title{Targeting pure quantum states by strong noncommutative
  dissipation}

\author{Vladislav Popkov} \thanks{E-mail: popkov@uni-bonn.de}
\affiliation{ HISKP, University of Bonn, Nussallee 14-16, 53115 Bonn,
  Germany.}  \affiliation{ Centro Interdipartimentale per lo studio di
  Dinamiche Complesse, Universit\`a di Firenze, via G.  Sansone 1,
  50019 Sesto Fiorentino, Italy } \author{Carlo Presilla}
\affiliation{Dipartimento di Fisica, Sapienza Universit\`a di Roma,
  Piazzale Aldo Moro 2, Roma 00185, Italy} \affiliation{Istituto
  Nazionale di Fisica Nucleare, Sezione di Roma 1, Roma 00185, Italy}
\author{Johannes Schmidt} \affiliation{Institut f\"{u}r Teoretische
  Physik, Universit\"{a}t zu K\"{o}ln, Z\"ulpicher Strasse 77,
  K\"{o}ln, Germany.}

\begin{abstract}
  We propose a solution to the problem of realizing a predefined and
  arbitrary pure quantum state, based on the simultaneous presence of
  coherent and dissipative dynamics, noncommuting on the target state
  and in the limit of strong dissipation.  More precisely, we obtain a
  necessary and sufficient criterion whereby the nonequilibrium steady
  state (NESS) of an open quantum system described by a Lindblad
  master equation approaches a target pure state in the Zeno regime,
  i.e., for infinitely large dissipative coupling.  We also provide an
  explicit formula for the characteristic dissipative strength beyond
  which the purity of the NESS becomes effective, thus paving the way
  to an experimental implementation of our criterion. For an
  illustration, we deal with targeting a Bell state, an arbitrary pure
  state of $N$ qubits, and a spin-helix state of $N$ qubits.
\end{abstract}

\date{\today }
\maketitle

\section{Introduction} Preparation of entangled states is a core
problem in most protocols of quantum information science and its
technological applications.  Target quantum states can be generated,
in principle, in various ways: via coherent dynamics, dissipative
dynamics, or a generic combination of both coherent and dissipative
dynamics~\cite{ShankarNature2013Bell,LinNature2013Bell,
  BlattPhysReports08,BlattNature2013,Kienzler2015Science,
  Kastoryano2011PRL,TicozziViola2012,PhysRevA.87.033802,
  PhysRevLett.107.080503,Diehl2011,PhysRevLett.106.020504,
  PhysRevLett.110.120402,PhysRevLett.110.253601,PhysRevLett.117.040501,
  1367-2630-14-6-063014,1367-2630-14-5-053022,
  PhysRevLett.111.246802,PhysRevA.88.023849,Barreiro2011}.

The combined dissipative generation of pure states, which best mimics
the typical experimental conditions, has the advantage of being stable
against decoherence and almost independent of the initial conditions.
Both these features are due to the existence of a nonequilibrium
steady state (NESS), assumed to be unique, of the Lindblad master
equation containing a non-unitary part that models a dissipative
coupling of a quantum system to the environment.  Apart from
limitations due to imperfections of experimental setups, an
\textit{ideal target} might be unreachable because of fundamental
limitations. An optimization problem then arises, namely, how to
approach the target maximally closely with the available resources.
The task can be framed in a more general context of problems of
optimal control of open quantum systems~\cite{Dirr2009} with vast
technological applications~\cite{Glaser2015}.

In the present context, the ideal target is to generate a chosen pure
NESS dissipatively, i.e., to obtain a pure state solution,
$\rho=\ket{\Psi}\bra{\Psi}$, of a stationary master equation of
Lindblad form~\cite{Lindblad,GKS}
\begin{align*}
  -\frac{i}{\hbar} \left[H,\rho\right] + \sum_{\alpha} \Gamma_{\alpha}
  \left( L_{\alpha} \rho L_{\alpha}^\dag - \frac{1}{2} \left(
      L_{\alpha}^\dag L_{\alpha} \rho + \rho L_{\alpha}^\dag
      L_{\alpha} \right) \right)=0.
\end{align*}
To accomplish this task the available resources are: the dissipative
actions (the Lindblad operators $L_{\alpha}$), the accessible ranges
of the respective dissipative strengths (the parameters
$\Gamma_{\alpha}$), and the accessible coherent evolutions (the
effective Hamiltonian $H$).  It is known~\cite{Yamamoto05,ZollerPRA08}
that an exact pure NESS of the above master equation requires a
\textit{commuting} action of the coherent and dissipative parts of the
dynamics on the target state.  More precisely, the target state $|
\Psi \rangle$ must satisfy
\begin{align}
  &H | \Psi \rangle = \lambda | \Psi \rangle, \label{Yamamoto1}\\
  &L_\alpha | \Psi \rangle =0, \qquad \mbox{ for all $\alpha$}.
  \label{Yamamoto2}
\end{align}
Often, Eqs.~(\ref{Yamamoto1}) and (\ref{Yamamoto2}) cannot be
simultaneously satisfied, making the ideal target unreachable.

Suppose, however, that one can still satisfy, or approximately
satisfy, some of Eqs.~(\ref{Yamamoto1}) and (\ref{Yamamoto2}). Which
conditions are crucial and which can be relaxed without ruining the
NESS? A universal answer to this question does not exist, as it would
require a perturbative analysis of a concrete Liouvillean operator.

Our aim is to demonstrate that one can generate an almost pure NESS by
a weaker criterion than the ``commutativity on the state'', at the
cost of increasing the strength of the dissipation. This strength can
somewhat be manipulated in experiments, e.g., a dissipative loss rate
of atoms in an optical trap is controlled by a laser beam intensity.
Even the Zeno limit, namely, the limit of infinitely large dissipative
couplings $\Gamma_\alpha$, can be reached
experimentally~\cite{ZenoStaticsExperimental, TonksgasScience,
  ZenoDynamicsExperimental, Signoles2014}, and it often produces
surprising
effects~\cite{Zeno,TonksgasNJP,TonksgasPRA2009,Zanardi2014PRL,Zanardi2015PRA}.
Another important aspect of our criterion is that, except for its
strength, the dissipation is regarded as fixed, whereas it is $H$ to
be considered adjustable, i.e., $H$ plays the role of the control
parameter providing the desidered nature of the target state
$\ket{\Psi}$.  Again, this point of view appears to be experimentally
approachable~\cite{2016OtteXXZ}.

According to our weaker criterion, coherent and dissipative actions on
a target state need not be commutative.  Violation of the
``commutativity on the state'' forbids, in a strict mathematical
sense, to generate an exact pure NESS. Instead, a state, which is
arbitrarily close to the exact one, up to a controlled error, can be
dissipatively generated.  Our criterion thus naturally consists of two
parts, given by following theorems~\ref{theorem1} and
\ref{theorem2}. Theorem~\ref{theorem1} introduces the weaker criterion
valid for \textit{strong} dissipation. The criterion essentially
consists of a condition which relates the effective Hamiltonian $H$,
the pure NESS $\ket{\Psi}$, and the dark state of the dissipator,
namely, the eigenvector corresponding to zero eigenvalue of the
Lindbladian in the limit of infinitely large dissipation.  In
formulating this criterion, we suppose that the dissipation acts only
on a subset of the degrees of freedom of the system, typically a small
subset, as in the case of boundary driven systems.  Theorem
\ref{theorem2} quantifies how strong the dissipation should be in
order to generate an almost pure NESS, and imposes further
(exceptional) restrictions on the Hamiltonian.  These restrictions
concern isolated points where the weak criterion alone fails or might
fail.  Thus, theorems \ref{theorem1} and \ref{theorem2} provide,
respectively, necessary and sufficient conditions for generating a
pure NESS in the Zeno limit.

To illustrate theorems \ref{theorem1} and \ref{theorem2}, we discuss
three examples: the generation of a Bell state of two qubits in a
system of three interacting qubits (see
Fig.~\ref{BellState_Targeting}), its $N$-qubit generalization, namely,
the generation of an arbitrary pure state of $N$ qubits in a system of
$N+1$ interacting qubits, and, finally the generation of a factorized
spin-helix state in a boundary driven spin chain with nearest-neighbor
interaction (see Fig.~\ref{SpinHelixState_Targeting}).  Note that, as
it will be clear from the precise statement of our criterion, in all
these example we consider an enlarged system to obtain the effectively
desired pure NESS, e.g., $N+1$ qubits for a state of $N$ qubits.  In
the first example, we illustrate the presence of exceptional points
where the criterion may fail.  In the second one, we explicitly show
how a NESS approaches a pure state in the limit of large dissipative
strength and evaluate the NESS relaxation time. With the third example
we demonstrate how, for a given set of resources, the ``commutativity
on the state'' condition is never fulfilled, while our weaker
condition can easily be met by tuning a model parameter. In all
examples, the most common type of qubit dissipation is considered,
namely, the polarization, with an adjustable degree of one or two
spins.
\begin{figure}[t]
  \begin{center}
    \includegraphics[width=0.61\columnwidth,clip]{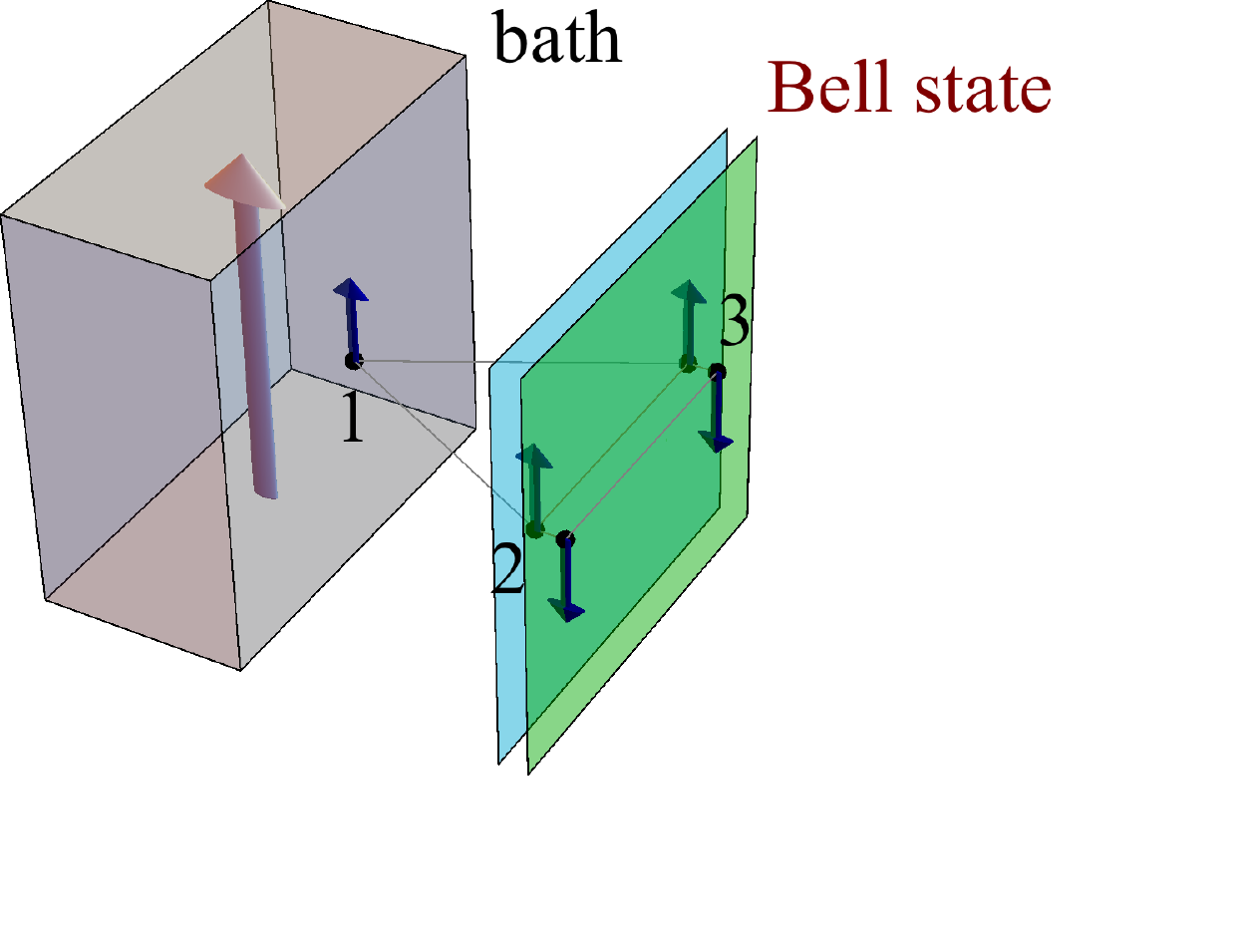}
    \caption{ Targeting a Bell state $\textstyle{\frac{1}{\sqrt{2}}}
      \left( \ket{\uparrow_2 \uparrow_3} + \ket{\downarrow_2
          \downarrow_3}\right)$ in a system of $N=3$ qubits by locally
      coupling qubit 1 to a fully polarizing bath.  Targeting
      criterion (\ref{ConditionTargetPure}) is satisfied by the
      Hamiltonian $H$ defined in
      Fig.~\ref{Fidelity&Gammach_vs_lambda}.}
    \label{BellState_Targeting}
  \end{center}
\end{figure}
\begin{figure}[t]
  \begin{center}
    \includegraphics[width=0.99\columnwidth]{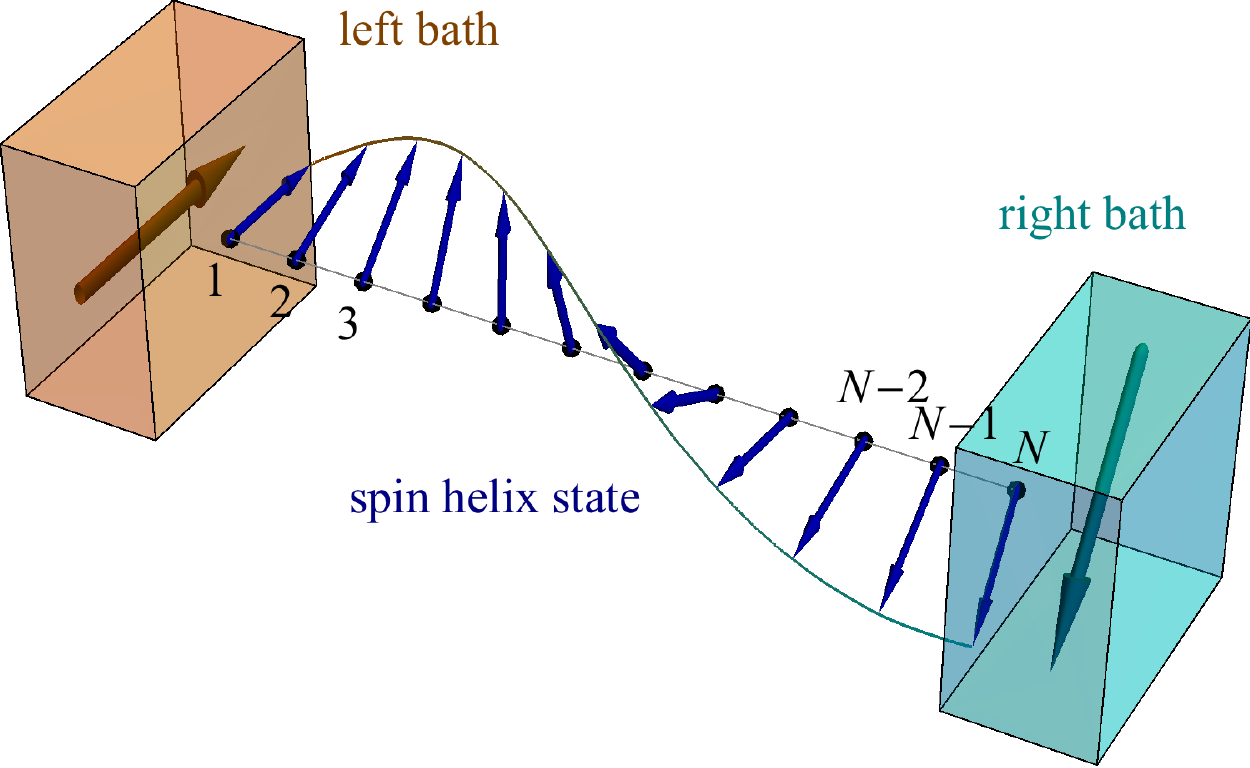}
  \end{center}
  \caption{ Targeting a spin-helix state in a chain of $N$ spins by
    locally coupling boundary spins $1$ and $N$ to two fully
    polarizing baths.  Targeting criterion (\ref{ConditionTargetPure})
    is satisfied by assuming $H$ as the Hamiltonian of the $XXZ$
    Heisenberg spin $1/2$ model.}%
  \label{SpinHelixState_Targeting}
\end{figure}

\section{Main results}
To formalize the problem, consider a finite quantum system in contact
with an external environment.  The time evolution of the reduced
density matrix of the system $\rho$ is described by a quantum master
equation in the Lindblad form
\cite{Petruccione,PlenioJumps,ClarkPriorMPA2010} (we set $\hbar=1$),
\begin{align}
  \frac{\partial\rho}{\partial t}=\mathcal{L}\rho= -i\left[
    H,\rho\right] +\Gamma \mathcal{D}_L \rho,
  \label{LME}
\end{align}
where $H$ represents the unitary part of the evolution of the reduced
density matrix, and $\mathcal{D}_L$ is the Lindblad dissipator.  For
simplicity, we suppose that the dissipator contains a single Lindblad
operator,
\begin{align}
  \mathcal{D}_L\rho = L \rho L^\dag- \frac{1}{2} \left( L^\dag L \rho
    + \rho L^\dag L \right),
  \label{DefDissipator}
\end{align}
with the jump operator $L$ acting in a subspace $\mathcal{H}_0$ of the
whole Hilbert space $\mathcal{H}$ of the system.  We also suppose,
this in an important point, that the kernel of $\mathcal{D}_L$ is
nondegenerate and pure in $\mathcal{H}_0$, i.e., it exists a unique
pure state of $\mathcal{H}_0$, indicated by $\ket{\psi_{\rm{Zeno}}}$,
such that $\mathcal{D}_L
\ket{\psi_{\rm{Zeno}}}\bra{\psi_{\rm{Zeno}}}=0$.  Note that we have $
\mathcal{H}= \mathcal{H}_0 \otimes \mathcal{H}_1$, with
$\mathcal{H}_0$, $\mathcal{H}_1$ having finite dimensions $d_0$ and
$d_1$, respectively. The state $\ket{\psi_{\rm{Zeno}}}$ should not be
confused with the desired target state, which we call
$\ket{\psi_{\rm{target}}}$ and which resides in $\mathcal{H}_1$.  With
an abuse of notation, which, however, should not generate troubles,
the product state $\ket{\Psi } = \ket{\psi_{\rm{Zeno}}} \otimes
\ket{\psi_{\rm{target}}}$ will also be designated as the target state
in the space $\mathcal{H}$.  As it will be clear from the statement of
the theorems given below, we use a part of the degrees of freedom of
the system, namely, $\mathcal{H}_0$, as an \textit{auxiliary} subspace
to achieve, in the Zeno limit, the desired target state in
$\mathcal{H}_1$. This constitutes our minimal setup.  Note that the
reduced density matrix $\rho$ which appears in Eqs.~(\ref{LME}) and
(\ref{DefDissipator}) belongs to the full Hilbert space $\mathcal{H}$.
Finally, we assume that the NESS, i.e., the time-independent solution
of Eq.~(\ref{LME}), that we denote as $\rho_{\rm{NESS}}(\Gamma)$, is
unique for any $\Gamma$.  The necessary and sufficient criterion for
the generation of a pure NESS $\ket{\Psi}$ in the Zeno limit are given
by the following theorems \ref{theorem1} and \ref{theorem2}.
\begin{theorem}[necessary condition]
  \label{theorem1}
  If $\lim_{\Gamma\to\infty} \rho_{\rm{NESS}}(\Gamma) =
  \ket{\Psi}\bra{\Psi}$, with $\ket{\Psi } = \ket{\psi_{\rm{Zeno}}}
  \otimes \ket{\psi_{\rm{target}}}$ and
  $\ket{\psi_{\rm{target}}}\in\mathcal{H}_1$, then the Hamiltonian $H$
  satisfies
  \begin{align}
    H \ket{\Psi} = \lambda \ket{\Psi} + \kappa
    \ket{\psi_{\rm{Zeno}}^\perp} \otimes \ket{\psi_{\rm{target}}},
    \label{ConditionTargetPure}
  \end{align}
  where $\braket{\psi_{\rm{Zeno}}}{\psi_{\rm{Zeno}}^\perp}=0$.
\end{theorem}
The proof of theorem \ref{theorem1} is given in
Appendix~\ref{SuppMat}.  Note that Eq.~(\ref{ConditionTargetPure})
entails the noncommutativity of the coherent and dissipative parts of
the dynamics on the state $\ket{\Psi}$, namely, $H L\ket{\Psi} \neq L
H\ket{\Psi}$.  Indeed, $H L\ket{\Psi}=0$, since $L\ket{\Psi}=0$ by the
Zeno regime assumption.  On the other hand, $ L H\ket{\Psi}=\kappa L
\ket{\psi_{\rm{Zeno}}^\perp} \otimes \ket{\psi_{\rm{target}}} \neq 0$,
for any nonzero $\kappa$.

Theorem 2 provides a sufficient condition for $\rho_{\rm{NESS}}$ to
converge to a pure target state in the Zeno limit and shows that this
state is approached algebraically, as $ 1-
\tr\rho_{\rm{NESS}}^2(\Gamma) = \left( \Gamma_{\rm{ch}}/\Gamma
\right)^2+ O(\Gamma^{-3}) $, everywhere, except at isolated points
where criterion~(\ref{ConditionTargetPure}) may break down.  The
measure $\epsilon(\Gamma) = 1- \tr\rho_{\rm{NESS}}^2(\Gamma)$ is
chosen as a convenient measure of purity as it is strictly positive
for a mixed state and vanishes if and only if $\rho_{\rm{NESS}}$ is
pure. Theorem 2 also gives an explicit expression for the
characteristic dissipative strength $\Gamma_{\rm{ch}}$ beyond which an
almost pure NESS can be reached.

To enunciate our result, we need to introduce some convenient
notations.  Let the vectors $\ket{e^j}$, $j=0,1,\dots,d_0-1$ form an
orthonormal basis in $\mathcal{H}_0$.  Without loss of generality, we
choose $\ket{e^0} \equiv \ket{\psi_{\rm{Zeno}}}$ and $\ket{e^{1}}
\equiv \ket{\psi_{\rm{Zeno}}^\perp}$.  Block-decompose $H$ in the
basis of $\ket{e^j} $ as
\begin{align}
  H = \sum_{j=0}^{d_0-1}\sum_{k=0}^{d_0-1} H_{jk}, \qquad H_{jk} =
  |e^{j}\rangle \langle e^k | \otimes h_{jk},
  \label{Hdecomposition}
\end{align}
where $h_{jk}= \bra{e^{j}} H \ket{e^{k}}$ is a Hamiltonian acting in
$\mathcal{H}_1$. Let the eigenvectors of $h_{00}$, $\ket{\al}$, where
$h_{00} \ket{\al}= \la_\al \ket{\al}$, form an orthonormal basis in
$\mathcal{H}_1$. Note that $h_{jk}^\dagger=h_{kj}$ due to the
Hermiticity of $H$.
\begin{theorem}[sufficient condition]
  \label{theorem2}
  Let the Hamiltonian $H$ satisfy the condition
  (\ref{ConditionTargetPure}) with $\ket{\psi_{\rm{target}}} \equiv
  \ket{0}$, $\lambda \equiv \lambda_0$, and $\kappa \neq 0$.  Then,
  $\lim_{\Gamma\to\infty} \rho_{\rm{NESS}}(\Gamma) =
  \ket{\Psi}\bra{\Psi}$, with $\ket{\Psi } = \ket{\psi_{\rm{Zeno}}}
  \otimes \ket{\psi_{\rm{target}}}$.  Moreover,
  $\lim_{\Gamma\to\infty} \Gamma^2 ( 1 - \tr
  \rho^2_\mathrm{NESS}(\Gamma) ) = \Gamma_{\rm{ch}}^2$, with
  \begin{align}
    {\Gamma_{\rm{ch}}}^2 =8 |\kappa|^2 \sum_{\alpha=1}^{d_1-1}
    \sum_{\beta=1}^{d_1-1} (K^{-1})_{\alpha \beta} R_\beta.
    \label{ResGammaCharacteristic}
  \end{align}
  The characteristic dissipative strength $\Gamma_{\rm{ch}}$ is
  expressed in terms of the inverse of the $(d_1-1) \times (d_1-1) $
  matrix $K$ with elements $K_{\alpha\beta}$,
  $\alpha,\beta=1,2,\dots,d_1-1$, given by
  \begin{align}
    K_{\alpha\beta} = \sum_{k=1}^{d_0-1} \left( \left| \bra{\alpha}
        h_{k0} \ket{\beta} \right|^2 - \delta_{\alpha,\beta}
      \bra{\alpha} h_{k0}^\dagger h_{k0} \ket{\alpha} \right),
    \label{DefK}
  \end{align}
  and the real numbers $R_\alpha = \bra{\alpha} F \ketbra{0}{0}
  F^{\dagger} \ket{\alpha}$, $\alpha=1,2,\dots,d_1-1$, where $ F =
  \sum_{k=1}^{d_0-1} \left( h_{k1} + \left[ \Lambda h_{01}, h_{k0}
    \right] \right)$, and
  \begin{align}
    \Lambda = \sum_{\alpha=1}^{d_1-1}
    \frac{1}{\lambda_\alpha-\lambda_0} | \alpha \rangle \langle \alpha
    |.
    \label{DefLambda}
  \end{align}
\end{theorem}

Criterion (\ref{ConditionTargetPure}) and
Eqs.~(\ref{ResGammaCharacteristic})-(\ref{DefLambda}) are our main
results. Actually, for a chosen $\ket{\psi_{\rm{target}}}$, they allow
one to find a system (to construct a Hamiltonian $H$) which, once
coupled to a dissipator of the form (\ref{DefDissipator}) via the
Lindblad Eq.~(\ref{LME}) with a dissipative strength $\Gamma \gg
\Gamma_\textrm{ch}$, admits the target state $\ket{\Psi } =
\ket{\psi_{\rm{Zeno}}} \otimes \ket{\psi_{\rm{target}}}$ as NESS.
Note that theorem \ref{theorem2} does not just provide the
characteristic dissipative strength, but imposes further restrictions
on $H$ by implicitly requiring a nondivergent $\Gamma_\textrm{ch}$,
see below.  For this reason, theorems \ref{theorem1} and
\ref{theorem2} can be viewed as a \textit{rough criterion} (theorem
\ref{theorem1}) and \textit{refined criterion} (theorem
\ref{theorem2}) for a pure NESS generation.

Two lemmas given in the Appendix~\ref{SuppMat} provide necessary the
technical background for the proof of the theorems. Lemma \ref{lemma1}
shows that the Lindblad operator $L$ is not arbitrary, but is actually
fixed, up to a similarity transformation, by the assumption that
$\mathcal{D}_L$ has a nondegenerate and pure kernel in
$\mathcal{H}_0$. The dissipator (\ref{DefDissipator}) is then
completely fixed.  In lemma \ref{lemma2}, we calculate the dissipator
inverse $\mathcal{D}_L^{-1}$, that appears in the recurrence relation
(\ref{RecurrenceRho}).  Lemma \ref{lemma2} also gives a condition for
the existence of the dissipator inverse, leading to the secular
conditions (\ref{Cond1}).

The main idea behind the demonstration of the theorems is relatively
simple.  We start assuming that there exists an expansion of the NESS
in powers of $1/\Gamma$,
\begin{align}
  \rho_{\rm{NESS}}(\Gamma) = \sum_{m=0}^\infty
  \frac{\rho^{(m)}}{\Gamma^m},
  \label{NessExpansion.1}
\end{align}
convergent for sufficiently large $\Gamma$, let us say
$1/\Gamma<1/\Gamma_\mathrm{cr}$. In general, we are not able to
evaluate the critical dissipative strength $\Gamma_\mathrm{cr}$ (not
to be confused with $\Gamma_{\rm{ch}}$), but we always suppose it to
be inside the convergence disk of the series.  Plugging the expansion
(\ref{NessExpansion.1}) into the Lindblad equation (\ref{LME}) and
comparing order by order, we get a series of conditions for the
existence of the terms $\rho^{(m)}$.  The conditions for the existence
of $\rho^{(1)}$ provide Eq.~(\ref{ConditionTargetPure}).  In theorem
\ref{theorem2}, viceversa, the hypothesis that $H$ satisfies the
criterion (\ref{ConditionTargetPure}) together with the conditions for
the existence of $\rho^{(m)}$ at any order $m$ permits one to conclude
that $\lim_{\Gamma\to\infty} \rho_{\rm{NESS}}(\Gamma) =
\ket{\psi_{\rm{Zeno}}}\bra{\psi_{\rm{Zeno}}} \otimes
\ket{\psi_{\rm{target}}}\bra{\psi_{\rm{target}}}$ and also to evaluate
$\lim_{\Gamma\to\infty} \Gamma^2 ( 1- \tr \rho^2_{\rm{NESS}}(\Gamma))
= \Gamma_{\rm{ch}}^2$.

The characteristic dissipative strength $\Gamma_{\rm{ch}}$ may diverge
at two different kinds of singularities: those associated with the
poles of $\Lambda$ in Eq.~(\ref{DefLambda}) and those arising from the
noninvertibility points of the matrix $K$ defined by Eq.~(\ref{DefK}).
Both kinds of singularities signal an inconsistency of the $1/\Gamma$
power expansion converging to the pure NESS in the Zeno limit.  At the
singular points, the NESS in the Zeno limit generically is not pure.
It may also happen that $\Gamma_{\rm{ch}}=0$. In this case, the speed
of convergence to zero of the purity $\epsilon(\Gamma)$ must be
evaluated at the next $1/\Gamma$ order.

The singularities connected with the poles of $\Lambda$ are the most
interesting (see Fig.~\ref{Fidelity&Gammach_vs_lambda}). They appear
when the eigenvalue $\lambda_0$ of $h_{00}$, corresponding to the
eigenstate $\ket{0}$, becomes degenerate. There can be at most $d_1-1$
singularity points $\lambda_0=\lambda_\beta$, $\beta=1,\ldots d_1-1$.
Occasionally, by vanishing of the matrix element
$\bra{\alpha}\left[\Lambda h_{01},h_{k0} \right] \ket{0} $, the
$\al$-th singularity can be removed.

On the other hand, singularities due to non-invertibility of the
matrix $K$ are rare since they occur only by simultaneous vanishing of
several entries of the Hamiltonian (see Appendix~\ref{SuppMat}).  We
demonstrate both kinds of singularities in the following example.

\begin{figure}[t]
  \begin{center}
    \includegraphics[width=0.99\columnwidth,clip]{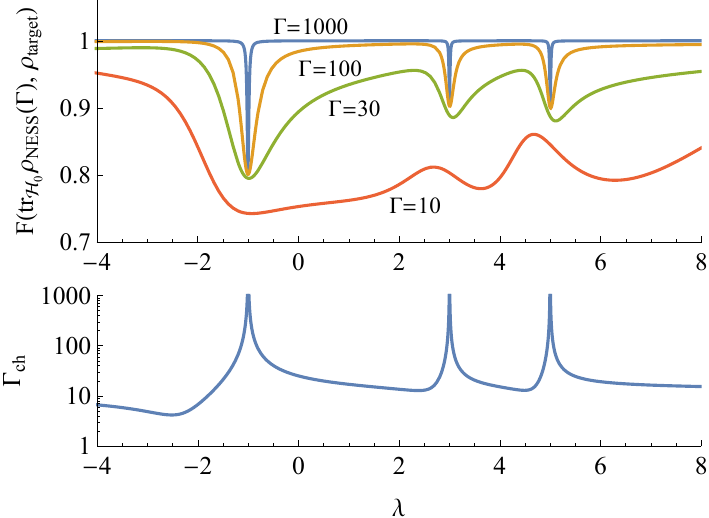}
    \caption{System of $N=3$ qubits with Lindblad operator
      $L=\sigma_1^\dag$ and target state
      $\ket{\psi_{\rm{Bell}}}=\textstyle{\frac{1}{\sqrt{2}}} \left(
        \ket{\uparrow_2 \uparrow_3} + \ket{\downarrow_2
          \downarrow_3}\right)$.  Bottom panel: characteristic
      dissipative strength $\Gamma_{\rm{ch}}$ vs $\lambda$. Top panel:
      fidelity $F\left( \traccazero\rho_{\rm{NESS}}(\Gamma),
        \rho_{\rm{target}} \right)$ vs $\lambda$ for
      $\Gamma=10,30,100,1000$.  For $\Gamma\gg \Gamma_{\rm{ch}}$, the
      NESS approaches the pure target state
      $\ket{\uparrow_1}\otimes\ket{\psi_{\rm{Bell}}}$ for all
      $\lambda$ except the degeneracy points $\lambda=\lambda_\alpha$,
      $\alpha=1,2,3$.  The Hamiltonian $H$ of the model is defined by
      specifying the matrix elements of the blocks
      $h_{00},h_{01},h_{10},h_{11}$ in the orthonormal basis
      $\{\ket{\alpha}\}$, where $\ket{0}=\ket{\psi_{\rm{Bell}}}$,
      $\ket{1}=\ket{\uparrow_2\downarrow_3}$,
      $\ket{2}=\ket{\downarrow_2\uparrow_3}$, $\ket{3}=
      \frac{1}{\sqrt{2}} \left( \ket{\uparrow_2 \uparrow_3} -
        \ket{\downarrow_2 \downarrow_3}\right)$.  We set
      $\bra{\alpha}h_{00}\ket{\beta}=\lambda_\alpha
      \delta_{\alpha,\beta}$, with $\lambda_0\equiv\lambda$,
      $\lambda_1=-1$, $\lambda_2=3$ and $\lambda_3=5$,
      $\bra{0}h_{10}\ket{0}=0.7$ and $\bra{\alpha}h_{10}\ket{0}=0$,
      with $\alpha=1,2,3$.  All the other matrix elements are assigned
      according to the following formula:
      $\bra{\alpha}h_{01}\ket{\beta}=Q_{\alpha+1,\beta+5}+\hc$,
      $\bra{\alpha}h_{10}\ket{\beta}=Q_{\alpha+5,\beta+1}+\hc$,
      $\bra{\alpha}h_{11}\ket{\beta}=\frac{1}{2}
      \left(Q_{\alpha+5,\beta+5}+\hc\right)$, where
      $Q_{k,j}=i^{k-j}\left(\mathrm{mod}\left(\left\lfloor 7 \tan(k^7
            j^4) \right\rfloor,2\right)+0.7\right)$.  Note that for
      each value of the parameter $\lambda$, we have a different
      Hamiltonian $H(\lambda)$. }
    \label{Fidelity&Gammach_vs_lambda}
  \end{center}
\end{figure}
\begin{figure}[t]
  \begin{center}
    \includegraphics[width=0.99\columnwidth,clip]{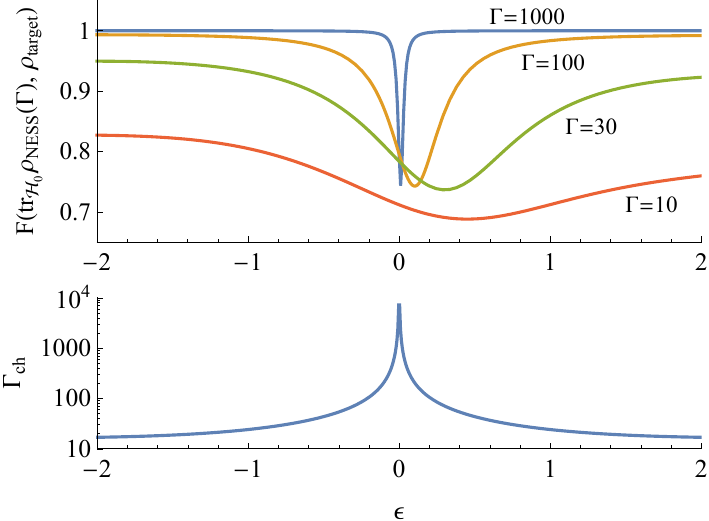}
    \caption{ As in Fig.~\ref{Fidelity&Gammach_vs_lambda}, but with
      the characteristic dissipative strength and the fidelity plotted
      as a function of the Hamiltonian parameter $\eps =
      \bra{0}h_{10}\ket{1}$. We have also set
      $\bra{0}h_{10}\ket{\al}=0$, for $\alpha=2,3$, while the rest of
      the Hamiltonian matrix elements are equal to those of
      Fig.~\ref{Fidelity&Gammach_vs_lambda} with $\lambda=1$.  For
      $\Gamma\gg \Gamma_{\rm{ch}}$, the NESS approaches the pure
      target state $\ket{\uparrow_1}\otimes\ket{\psi_{\rm{Bell}}}$ for
      all $\eps$ except the singular point $\eps=0$, where the matrix
      $K$ of Eq.~(\ref{DefK}) is not invertible. }
    \label{Fidelity&Gammach_vs_epsilon}
  \end{center}
\end{figure}

\section{Two-qubit Bell state}
As the simplest nontrivial example, still interesting from an
applicative point of view, we target a Bell state of two qubits.  As
shown in Fig.~\ref{BellState_Targeting}, to realize this state we
consider a system of $N=3$ qubits and a Lindblad operator acting on
the first qubit only, $L=\sigma^+ \otimes I \otimes I$.  Note that a
Bell state of two qubits cannot be a dark state of the operators
$\sigma^{+} \otimes I$ or $I \otimes \sigma^{+}$.  Therefore, for a
dissipative generation of a Bell state of two qubits in the Zeno
regime at least three qubits are needed. The degrees of freedom of one
qubit, that we designate as qubit 1, play the role of the auxiliary
subspace $\mathcal{H}_0$.  In the Hilbert subspace $\mathcal{H}_0
\equiv \mathbb{C}_2$, we choose the standard basis vectors $\{|e^{0}
\rangle,|e^{1} \rangle\}= \{\ket{\uparrow},\ket{\downarrow}\}$.  Note
that $\ket{e^0} \equiv \ket{\psi_{\rm{Zeno}}}$ is the unique pure
state of $\mathcal{H}_0$ such that $\mathcal{D}_L
\ket{\psi_{\rm{Zeno}}}\bra{\psi_{\rm{Zeno}}}=0$. We also have,
necessarily, in view of the dimension $d_0=2$, $\ket{e^1} \equiv
\ket{\psi_{\rm{Zeno}}^\perp}$.  In the complement Hilbert subspace
$\mathcal{H}_1 \equiv {\mathbb{C}_2} \otimes \mathbb{C}_2$, the target
state reads
\begin{align}
  \ket{\psi_\mathrm{Bell}}=\frac{1}{\sqrt{2}} \left( \ket{\uparrow_2
      \uparrow_3} + \ket{\downarrow_2 \downarrow_3}\right),
\end{align}
where subscripts denote the embedding.  To make sure that the Lindblad
Eq.~(\ref{LME}) admits, in the Zeno limit, the pure NESS
$\ket{\uparrow}\bra{\uparrow} \otimes
\ket{\psi_{\rm{Bell}}}\bra{\psi_{\rm{Bell}}}$, we just have to find a
Hamiltonian $H$ that satisfies the
criterion~(\ref{ConditionTargetPure}) with $\ket{\psi_{\rm{Bell}}}
\equiv \ket{0}$, $\lambda \equiv \lambda_0$, and $\kappa \neq 0$.
Specifically, this amounts to determining the matrix elements of the
blocks $h_{00},h_{01},h_{10},h_{11}$ of $H$ in the orthonormal basis
$\{\ket{\alpha}\}_{\alpha=0}^{3}$, constituted by the eigenvectors of
$h_{00}$, $h_{00} \ket{\alpha} = \lambda_a \ket{\alpha}$, with the
constraint $\ket{0}=\ket{\psi_{\rm{Bell}}}$.  For simplicity, let us,
first, choose the basis of $\mathcal{H}_1$ given by
$\ket{0}=\ket{\psi_{\rm{Bell}}}$,
$\ket{1}=\ket{\uparrow_2\downarrow_3}$,
$\ket{2}=\ket{\downarrow_2\uparrow_3}$, and $\ket{3}=
\frac{1}{\sqrt{2}} \left( \ket{\uparrow_2 \uparrow_3} -
  \ket{\downarrow_2 \downarrow_3}\right)$, and then define $h_{00}$ in
this basis as the diagonal matrix
$h_{00}=\mathrm{diag}(\lambda_0,\lambda_1,\lambda_2,\lambda_3)$. The
criterion~(\ref{ConditionTargetPure}) just imposes the following other
conditions: $\bra{0}h_{10}\ket{0}=\kappa$ and
$\bra{\alpha}h_{10}\ket{0}=0$, $\alpha=1,2,3$. For the rest, besides
being Hermitian, the Hamiltonian $H$ can be chosen in an arbitrary
way, as long as the NESS remains unique.

In Fig.~\ref{Fidelity&Gammach_vs_lambda}, with $H$ chosen as specified
in the caption, we show, as a function of the parameter $\lambda
\equiv \lambda_0$, the behavior of the fidelity between the target
state $\rho_{\rm{target}}=
\ket{\psi_{\rm{Bell}}}\bra{\psi_{\rm{Bell}}}$ and the stationary
solution of Eq.~(\ref{LME}) traced in $\mathcal{H}_0$, with
$\rho_{\rm{NESS}}(\Gamma)$ evaluated numerically for different values
of $\Gamma$.  In the bottom panel, for the same values of $\lambda$,
we display the characteristic dissipative strength $\Gamma_{\rm{ch}}$
calculated according to Eq.~(\ref{ResGammaCharacteristic}).  It is
evident that, whenever $\Gamma \gg \Gamma_{\rm{ch}}$, the fidelity is
close to 1, i.e., $\traccazero \rho_{\rm{NESS}}(\Gamma)$ approaches
the pure target state.  This is true for any $\lambda$ except the
degeneracy points $\lambda=\lambda_\alpha$, $\alpha=1,2,3$, where the
matrix $\Lambda$ of Eq.~(\ref{DefLambda}) is singular and
$\Gamma_{\rm{ch}}$ diverges. At these points, $\lim_{\Gamma\to\infty}
\rho_{\rm{NESS}}(\Gamma)$ exists but is not a pure state; therefore
the fidelity remains different from 1 for arbitrarily large $\Gamma$.

Consider the Hamiltonian $H$ obtained from that of
Fig.~\ref{Fidelity&Gammach_vs_lambda} with $\lambda=1$, by zeroing all
the elements $\bra{0}h_{10}\ket{\al}=0$, for $\al>1$, and putting
$\bra{0}h_{10}\ket{1}=\eps$. If $\eps=0$, the matrix $K$ evaluated as
prescribed by Eq.~(\ref{DefK}) is a stochastic matrix and $\det K=0$
as a consequence of the Perron-Frobenius theorem.  For $\eps\neq 0$,
$\det K \neq 0$ and $K$ is invertible. Then, according to
Eq.~(\ref{ResGammaCharacteristic}), we expect a singularity at the
point $\eps=0$, where $K^{-1}$ does not exist. This is exactly the
singular behavior which emerges for large $\Gamma$ in
Fig.~\ref{Fidelity&Gammach_vs_epsilon} at $\eps=0$ .

\section{Arbitrary $\bm{N}$-qubit pure state}
It is possible to generalize the ``Bell state'' example provided for a
system with $2+1$ qubits, to an arbitrary target pure state in a
system of $N+1$ qubits.  Dissipation is made to act only on one qubit,
say qubit 1, which defines an auxiliary space and allows the reaching
of the target state for the remaining $N$ qubits.  A \textit{minimal}
Hamiltonian realizing this scenario is given below. For this minimal
Hamiltonian, we are able to calculate the NESS analytically for
arbitrary dissipative strengths $\Gamma$.

Let us denote by $\ket{e^{0}0}\equiv\ket{e^0}\otimes\ket{0}$ a pure
state to be targeted in the Zeno limit.  Notations are as above,
namely, $\ket{e^0}=\ket{\psi_\mathrm{Zeno}}$ and
$\ket{0}=\ket{\psi_\mathrm{target}}$.  The evolution of the density
matrix is described by Eq.~(\ref{LME}) with one local Lindblad
operator $L=\sigma_1^{+}$ acting
on the auxiliary spin $1$ only.  The Lindblad operator $L$ projects
the first spin onto the state $\ket{e^{0}}\equiv\ket{\uparrow}$, while
$\ket{e^{1}} \equiv \ket{\downarrow}$ is the vector completing the
basis of $\mathcal{H}_0 = \mathbb{C}_2$.  The state $\ket{0}$ is an
arbitrary vector in $\mathcal{H}_1 =
\left(\mathbb{C}_2\right)^{\otimes_N}$, the Hilbert space of the
remaining $N$ spins.

A simple Hamiltonian satisfying Eqs.~(\ref{ConditionTargetPure}) and
(\ref{Cond Invertibility}) has the form
\begin{align}
  H =& \sum_{\al=0}^{2^{N}-1} \left( \lambda_\alpha
    \ket{e^{0}\al}\bra{e^{0}\al} + d_\al \ket{e^{1}\al}\bra{e^{1}\al}
  \right) \nonumber\\ &+
  \left( \kappa\ket{e^{0}0}\bra{e^{1}0}+\hc\right)\nonumber\\
  &+\sum_{\al=0}^{2^{N}-2} \left(
    \eta_{\al}\ket{e^{0}(\al+1)}\bra{e^{1}\al}+\hc \right)
  \nonumber\\ &+\sum_{\al=1}^{2^{N}-2} \sum_{\be>\al}^{2^{N}-2}
  \left( d_{\al\be}\ket{e^{1}\al}\bra{e^{1}\be}+\hc\right),
  \label{ManyQubitsHamiltonian}%
\end{align}
where $\hc$ denotes the Hermitian conjugate. The terms proportional to
$\lambda_0$ and $\kappa$ guarantee condition
(\ref{ConditionTargetPure}) while the terms proportional to
$\eta_{\al}$ lift up the degeneracy of the Liouvillean. (Both $\kappa$
and all $\eta_{\al}$ are required to be nonzero.) All the other
parameters are free parameters.  The pure state $\ket{e^{0}0}$ is
exactly realized as NESS only in the Zeno limit.

It can be checked that the recurrence conditions (\ref{RecurrenceRho})
for the Hamiltonian (\ref{ManyQubitsHamiltonian}) simplify to
\begin{align}
  \rho^{\left(m+2\right)} =
  -\frac{4\zeta}{(\lambda_0-\lambda_1)^{2}}\rho^{\left(m\right)},
\end{align}
where
\begin{align}
  \rho^{\left(1\right)} = 2\mathrm{i}\kappa
  \left(\frac{\eta_{0}^{*}}{\lambda_0-\lambda_1}
    \ket{e^{0}0}\bra{e^{0}1}+\ket{e^{0}0}\bra{e^{1}0}\right) +\hc,
\end{align}
\begin{widetext}
  \begin{align}
    \rho^{\left(2\right)}=&
    4\left|\kappa\right|^{2}\left(-\frac{\left|\kappa\right|^{2}+
        \left|\eta_{0}\right|^{2}+(\lambda_0-\lambda_1)^{2}}
      {(\lambda_0-\lambda_1)^{2}} \ket{e^{0}0}\bra{e^{0}0} +
      \frac{\left|\kappa\right|^{2}+\left|\eta_{0}\right|^{2}}
      {(\lambda_0-\lambda_1)^{2}} \ket{e^{0}1}\bra{e^{0}1}
      +\ket{e^{1}0}\bra{e^{1}0} \right) \nonumber\\
    &+
    4\kappa\left(\eta_{0}^{*}\frac{\lambda_0^{2}-\left|\kappa\right|^{2}
        -\left|\eta_{0}\right|^{2}-\lambda_0\lambda_1-d_0
        (\lambda_0-\lambda_1)}{(\lambda_0-\lambda_1)^{2}}
      \ket{e^{0}0}\bra{e^{0}1}
      +\frac{\kappa^{*}\eta_{0}}{\lambda_0-\lambda_1}\ket{e^{0}1}
      \bra{e^{1}0}
    \right.\nonumber\\
    &\left.\phantom{++ 4\kappa}+
      \frac{\lambda_0^{2}-\left|\eta_{0}\right|^{2}-\lambda_0\lambda_1
        -d_0(\lambda_0-\lambda_1)}{\lambda_0-\lambda_1}
      \ket{e^{0}0}\bra{e^{1}0}+\hc\right),
  \end{align}
\end{widetext}
and
\begin{align}
  \zeta =&
  \left(\lambda_0^{2}-\lambda_0\lambda_1-\left|\eta_{0}\right|^{2}
    -d_0(\lambda_0-\lambda_1) \right)^{2} \nonumber\\
  &+\left|\kappa\right|^{2}\left( |\kappa|^{2}+2\left((
      \lambda_0-\lambda_1)^{2} + |\eta_{0}|^{2}\right)\right).
\end{align}
As a result, the NESS can be written as
\begin{align}
  \rho_\mathrm{NESS} =& \ket{e^{0}0}\bra{e^{0}0} \nonumber\\ &+
  \frac{1}{\Gamma^2} \sum_{n=0}^{\infty}\left(\frac{-4\zeta}
    {(\lambda_0-\lambda_1)^{2}\Gamma^2}\right)^{n}
  \left(\Gamma\rho^{\left(1\right)}+\rho^{\left(2\right)}\right).
  \label{ExactNESSexpansion}%
\end{align}
Summing up the geometrical series, we obtain
\begin{align}
  \rho_\mathrm{NESS} = \ket{e^{0}0}\bra{e^{0}0} +
  \frac{(\lambda_0-\lambda_1)^{2}}
  {4\zeta+(\lambda_0-\lambda_1)^{2}\Gamma^{2}}
  \left(\Gamma\rho^{\left(1\right)}+\rho^{\left(2\right)}\right).
  \label{ExactNESSsummed}
\end{align}
Formally, the Taylor expansion (\ref{ExactNESSexpansion}) converges
for $\Gamma>
\Gamma_\mathrm{cr}=2\sqrt{\zeta}/|\lambda_0-\lambda_1|$. However, the
analytical formula (\ref{ExactNESSsummed}) has no singularities and is
valid for arbitrary $\Gamma$.  Note that the NESS does not depend on
terms proportional to $d_\al$ for $\al\geq 1$ and $d_{\al \be}$.

\begin{figure}[t]
  \begin{center}
    \includegraphics[width=0.99\columnwidth,clip]{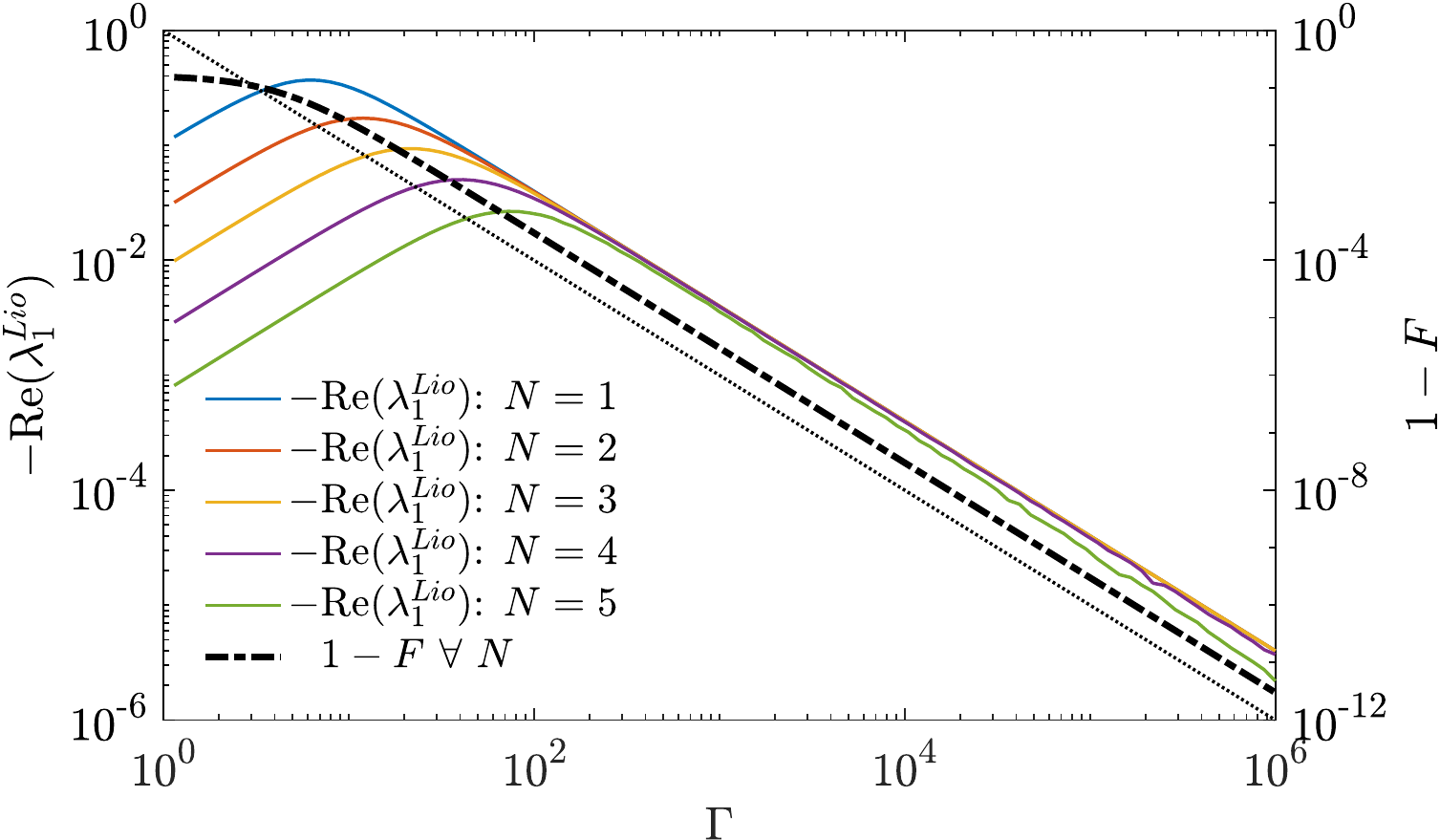}
    \caption{ Complement to 1 of the fidelity of the NESS (dot-dashed
      bold line, right vertical axis) and minus the real part of the
      Liouvillian gap versus $\Gamma$ for $N=1,2,3,4,5$ (solid lines
      from top to bottom, left vertical axis).  Note that $1-F$ is
      independent of $N$, and so become, for large $\Gamma$, the
      quantities $-\mathrm{Re}(\lambda_1^{Lio})$.  Parameters:
      $\lambda_{\alpha}=1+\alpha+\sqrt{\alpha},\kappa=1,\eta_{\alpha}=1\
      \forall\alpha, d_{\alpha}=0\ \forall\alpha,
      d_{\alpha\beta}=0$. The tiny dotted line, given by $y=1/x$ for
      the left $y$ axis, and by $y=1/x^2$ for the right $y$ axis, is a
      guide for the eyes.  }
    \label{Fig_Fidelity&Gap_vs_Gamma}
  \end{center}
\end{figure}

The purity measure $1-\tr\rho^2_\mathrm{NESS}(\Gamma)$ of the NESS for
finite dissipative strength is easily calculated.  The Taylor
expansion for large $\Gamma$ of formula (\ref{ExactNESSsummed}) gives
\begin{align}
  1-\tr\left(\rho_\mathrm{NESS}^{2}\right) =
  \frac{8\left|\kappa\right|^{4}}{(\lambda_0-\lambda_1)^{2}\Gamma^{2}}
  +\mathcal{O}\left(\Gamma^{-3}\right).
  \label{Tr(r-rr)INF}%
\end{align}
From this expression, we read off the characteristic dissipative
strength
\begin{align}
  \Gamma_\mathrm{ch} = \frac{ \sqrt{8} \left|\kappa\right|^{2}}
  {\left|\la_0-\la_1\right|}.
  \label{ResGammaCharacteristicExample3}%
\end{align}
This result can also be obtained directly by using
Eq.~(\ref{ResGammaCharacteristic}).  Note that
Eq.~(\ref{ResGammaCharacteristicExample3}) displays a divergence for
$\lambda_1=\lambda_0$, but not for other points of degeneracy of
$h_{00}$, i.e., $\lambda_\alpha=\lambda_0$, $\alpha>1$. This happens
because the matrix elements $\langle\alpha|h_{01}|0\rangle$ vanish for
all $\alpha>1$, so that system of Eqs.~(\ref{LinearSystemForM1Phi})
leads to a divergence only for $\alpha=1$ for which
$\langle\alpha|h_{01}|0\rangle \neq 0$.

Finally, the fidelity of the NESS with respect to the ideal target
state $\ket{e^0 0}$, namely, $F= \sqrt{\bra{e^0 0} \rho_\mathrm{NESS}
  \ket{e^0 0} }$, can be explicitly calculated from
Eq.~(\ref{ExactNESSsummed}) to be
\begin{align}
  F &=\sqrt{1-\frac{4\left|\kappa\right|^{2}
      \left( \left(\lambda_{0}-\lambda_{1}\right)^{2}+\left|\kappa\right|^{2}+
        \left|\eta_{0}\right|^{2}\right) }{4\zeta+
      \left(\lambda_{0}-\lambda_{1}\right)^{2}\Gamma^{2}}},
  \label{FidelityExample3}
\end{align}
while the Taylor expansion of $1-F$ for large $\Gamma$ is given by

\begin{align}
  1-F&=\frac{2\left|\kappa\right|^{2}
    \left( \left(\lambda_{0}-\lambda_{1}\right)^{2}+\left|\kappa\right|^{2}+
      \left|\eta_{0}\right|^{2}\right) }
  {\left(\lambda_{0}-\lambda_{1}\right)^{2}\Gamma^{2}}
  +\mathcal{O}\left(\Gamma^{-4}\right).
\end{align}

Figure~\ref{Fig_Fidelity&Gap_vs_Gamma} shows the fidelity $F$ of the
NESS, actually, its complement to 1, and the inverse of the system
relaxation time $\tau_\mathrm{relax}$, namely, minus the real part of
the Liouvillian gap, as a function of $\Gamma$ for different (small)
system sizes.  We observe that $1-F\sim \Gamma^{-2}$ and
$\tau_\mathrm{relax}^{-1} \sim \Gamma^{-1}$, at least asymptotically
in $\Gamma$. We conclude that the inverse relaxation time grows
sublinearly with $1-F$, i.e.,
\begin{align}
  \tau_\mathrm{relax} &\sim \left(1-F \right)^{-\frac{1}{2}},
\end{align}
for $1-F$ small enough.  Remarkably, both the asymptotic fidelity and
the asymptotic relaxation time do not depend on the system size $N$.

\section{Factorized spin-helix state of $\bm{N}$ qubits}
Our third example concerns targeting a factorized spin-helix pure
state in a chain of $N$ qubits (spins 1/2), which, in the absence of
dissipation, are supposed to evolve according to the paradigmatic
Heisenberg Hamiltonian.  The spin-helix state in a chain of $N$ qubits
has the form
\begin{align}
  \ket{\Psi} = \bigotimes_{k=1}^N \left(
    \begin{array}{c}
      \cos(\frac{\theta}{2})  e^{-i\varphi_k/2}
      \\
      \sin(\frac{\theta}{2})  e^{i\varphi_k/2}
    \end{array}
  \right), \quad \varphi_k = \gamma(k-1).
  \label{XXZtargetedstate}
\end{align}
This state describes a precession of the local spin along the chain
with fixed orbital angle $\theta$ and polar angle $\varphi$
homogeneously increasing as $\varphi_{k+1}-\varphi_{k}= \gamma$. The
boundary conditions $\theta_1=\theta, \varphi_1=0$ and
$\theta_N=\theta, \varphi_N=\gamma(N-1)=\phi$ can be realized by
coupling the first and last spins of the chain to two fully polarizing
baths.  In other words, we can consider a Lindblad master equation
like (\ref{LME}) with two dissipators, $\mathcal{D}_{L_L}$ and
$\mathcal{D}_{L_R}$, of the form (\ref{DefDissipator}), associated
with the left and right Lindblad operators, $L_L$ and $L_R$, acting
only on spin 1 and spin $N$, respectively. See
Fig.~\ref{SpinHelixState_Targeting} for an example with
$\theta=\pi/2$. More details on the setting of the model and its
properties in various limits can be found
in~\cite{XYweak,XYtwist,MPA,MPA-PRE2013,ProsenExact2011}.

A possible protocol leading to the evolution (\ref{LME}) with local
Lindblad operators $\mathcal{D}_{L_L}$, $\mathcal{D}_{L_R}$ can be
found in the Appendix A of~\cite{Landi}.  Despite the introduction of
a second dissipator, for any chain of $N>3$ spins with first-neighbor
interaction, $L_L$ and $L_R$ operate independently. It can be shown
that the conditions for the existence of the terms $\rho^{(m)}$ in the
expansion~(\ref{NessExpansion.1}) have exactly the same form of
Eq.~(\ref{Cond1}) in Appendix~\ref{SuppMat}, provided that now we
intend $\mathcal{H}_0 \equiv \mathbb{C}_2 \otimes \mathbb{C}_2 $.  We
conclude that our theorems \ref{theorem1} and \ref{theorem2} apply
unchanged.  We find that the criterion~(\ref{ConditionTargetPure}) for
the spin-helix state~(\ref{XXZtargetedstate}) with $\lambda=0$ and
$\kappa=-i \sqrt{2} \sin \theta \sin \gamma$ is satisfied by the
celebrated $XXZ$ Heisenberg Hamiltonian,
\begin{align}
  H_{XXZ}= \sum_{j=1}^{N-1} \left( \sigma_{j}^{x}\sigma_{j+1}^{x}+
    \sigma_{j}^{y}\sigma_{j+1}^{y}+ \cos\gamma
    \sigma_{j}^{z}\sigma_{j+1}^{z} \right).
  \label{HamiltonianXXZ}
\end{align}
Note that the $Z$-axis anisotropy is tuned to $\Delta=\cos \gamma$.
For $\theta=\pi/2$, convergence to a spin-helix state in the Zeno
limit was suggested in~\cite{PhysRevA.93.022111} via a different
method.  The factorized state~(\ref{XXZtargetedstate}) is a rather
unusual NESS for a quantum many-body driven system.  From a quantum
transport viewpoint, it carries a ballistic, system-size independent
magnetization current. It is characterized via a macroscopic winding
number, which counts the number of full rotations of the local
magnetization vector along the chain.  Moreover, since the actions of
the Lindblad operators and of the Hamiltonian on the
state~(\ref{XXZtargetedstate}) do not commute, the spin-helix state
cannot be obtained within the strong criterion of
Eqs. (\ref{Yamamoto1}) and (\ref{Yamamoto2}). On the other hand, by
our weaker criterion the spin-helix state can be spotted and
approached with arbitrary fidelity for sufficiently large dissipation.
Further elaborate analysis shows that for any angle $\gamma$ which is
a rational of $\pi$, $\gamma/\pi=n/m$, with integer $n,m$, convergence
to a spin-helix state breaks down for systems of size $N\geq
m+1$. This rather surprisingly behavior is an involved prediction of
present theorem~\ref{theorem2}. Details can be found
in~\cite{PSP2017}.  Note that, even though it generically takes longer
for a system with local Liouvillean to relax to the NESS, the typical
relaxation time is expected to grow only polynomially with the system
size, see~\cite{Znidaric2015}.

Similarly to~\cite{2016ZnidarichTargeting}, the spin-helix state
(\ref{XXZtargetedstate}) can be used to dissipatively generate an
arbitrary pure single-spin state on a remote location, arbitrarily far
from the place where the dissipation acts.  An experimental
verification of our results should be within the
reach~\footnote{A. F. Otte, private communication} of modern
experimental technics, which allow to build and probe, with
single-site resolution, arrays of magnetic atoms on a surface by
low-temperature scanning tunneling microscopy~\cite{2016OtteXXZ}.

\section{Conclusions}
To conclude, we have obtained a general criterion for a nonequilibrium
steady state to become a pure state of the form
$\ket{\psi_\mathrm{Zeno}}\otimes\ket{\psi_\mathrm{target}}$ in the
limit of strong dissipation, and calculated the relevant
characteristic dissipation strength.  The relative merit of our
approach is the absence of any strict requirements on the Lindblad
operators, besides the demand for locally targeting some pure state
$\ket{\psi_\mathrm{Zeno}}$. This is easily achieved by most commonly
implemented dissipative mechanisms like pumping, local dissipative
loss or particles from a trap. The key condition to obtain a desired
$\ket{\psi_\mathrm{target}}$ is then the restriction
(\ref{ConditionTargetPure}) imposed on the effective Hamiltonian. This
becomes a sufficient criterion provided the exceptions stated by
theorem \ref{theorem2} are circumvented.

\begin{acknowledgments}
  VP thanks the DFG for financial support. VP thanks G.~Sch\"utz and
  C. Kollath for useful discussions. We thank G. Jona-Lasinio for a
  critical reading.
\end{acknowledgments}

\appendix

\section{Proofs of theorems \ref{theorem1} and \ref{theorem2}}
\label{SuppMat}
We start with two lemmas, which are key for the following
demonstrations.

\begin{lemma}
  \label{lemma1}
  Let $\mathcal{D}_L$ be a Lindblad dissipator of the form
  $\mathcal{D}_L\rho = L \rho L^\dag- \frac{1}{2} \left( L^\dag L \rho
    + \rho L^\dag L \right)$ with $\rho$ density matrix operator in
  $\mathcal{H}$ and $L$ acting in the $d_0$-dimensional subspace
  $\mathcal{H}_0$ of the Hilbert space $\mathcal{H}$. If
  $\mathcal{D}_L$ has a nondegenerate pure kernel, i.e., there exists
  a unique pure state $\ket{\phi}$ such that $\mathcal{D}_L
  \ketbra{\phi}{\phi}=0$, then $L$ can be brought into a single Jordan
  block acting in $\mathcal{H}_0$, namely, $L=J_{0,d_0}$, where
  $\bra{e^i}J_{0,d_0}\ket{e^j}=\delta_{i,j-1}$, $i,j=0,1,\dots,d_0-1$.
\end{lemma}
\begin{proof}[Proof of lemma \ref{lemma1}]
  From the condition $\mathcal{D}_{L} | \phi\rangle \langle \phi|=0$,
  it follows that $ L \ket{\phi}= b \ket{\phi}$.  Without loss of
  generality, we can put $b=0$, i.e., $\ket{\phi}$ can always be made
  a dark state of $L$, by repartitioning the coherent and dissipative
  parts, $H$ and $L$, in the Lindblad master equation
  $\partial\rho/\partial t = -i\left[ H,\rho\right] +\Gamma
  \mathcal{D}_L \rho$, see \cite{Yamamoto05,ZollerPRA08}.  Because of
  the nondegeneracy assumption, no other eigenstates of $ L $ exist
  different from $| \phi\rangle$.  As a consequence, $L$ must be
  nondiagonalizable, meaning it is splitted into Jordan blocks by a
  similarity transformation, $V L V^{-1}= {\tilde L}$.  Each Jordan
  block separately would give rise to a state $| {\tilde \phi}\rangle$
  such that ${\tilde L} | {\tilde \phi}\rangle =0$, which would
  contradict the nondegeneracy assumption.  We conclude that there
  must be one Jordan block only. This is exactly the definition of
  $J_{0,d_0}$, up to an arbitrary scalar factor, ${\tilde L}= a
  J_{0,d_0}$. Without loss of generality, we set $L=J_{0,d_0}$, since
  the Hamiltonian is not defined yet and the coefficient $a$ is
  absorbed into the dissipative strength $\Gamma$.
\end{proof}

\begin{lemma}
  \label{lemma2}
  Let $\mathcal{D}_L$ be a Lindblad dissipator of the form specified
  in lemma~\ref{lemma1}
  and $X$ an operator acting in $\mathcal{H}$.  A necessary condition
  for $\mathcal{D}_L^{-1}X$ to exist is $\traccazero X=0$, where
  $\traccazero$ denotes the partial trace in $\mathcal{H}_0$.  If
  $L=J_{0,d_0}$, this condition is also sufficient.
\end{lemma}
\begin{proof}[Proof of lemma \ref{lemma2}]
  Suppose that $\mathcal{D}_L^{-1}X=Y$ exists.  Then $X =
  \mathcal{D}_L Y$ and, using the cyclic invariance of the trace, we
  have $\traccazero X = \traccazero\mathcal{D}_L Y = 0$.  Assume, now,
  that $L=J_{0,d_0}$.  Let us define a representation of a spin
  $s=(d_0-1)/2$ in the Hilbert space $\mathcal{H}_0$ of dimension
  $d_0$ and consider the basis $\ket{n}$, $n=0,1,\dots ,2s$, such that
  $J_{0,d_0} \ket{n} = \ket{n-1}$, $J_{0,d_0}\ket{0}=0$. Then, the
  dissipator ${\mathcal{D}}_{L}$ acts in the space spanned by
  $\ket{n}\bra{m}$ as
  \begin{widetext}
    \begin{align} \mathcal{D}_{L} \ket{n}\bra{m} = \left\{
        \begin{array}{ll}
          \ket{n-1}\bra{m-1}-\ket{n}\bra{m},\qquad &n>0,m>0, \\
          -\frac{1}{2} \ket{n}\bra{0}, & n>0,m=0, \\
          -\frac{1}{2} \ket{0}\bra{m}, & m>0,n=0, \\
          0, &n=m=0.
        \end{array}
      \right.
      \label{DissipatorSplusAction}
    \end{align}
    From Eq.~(\ref{DissipatorSplusAction}), it follows that for any
    matrix $X$ with zero trace, $\tr X=0$, we can construct
    $\mathcal{D}_{L}^{-1}X$. Indeed, since the dissipator
    $\mathcal{D}_{L}$ is a linear map, it is enough to show that every
    nondiagonal element of $X$ is invertible, and so is also the
    diagonal of $X$.  An arbitrary nondiagonal element
    $\ket{n}\bra{m}$, $n \neq m$, is readily inverted using
    Eq.~(\ref{DissipatorSplusAction}) as
    \begin{align}
      \mathcal{D}_{L}^{-1} \ket{n}\bra{m} = \left\{
        \begin{array}{ll}
          \sum_{k=0}^{n-1} \ket{n-k}\bra{m-k} -
          \frac{1}{2} \ket{0}\bra{m-n}, \qquad &
          n<m,
          \\
          \sum_{k=0}^{m-1} \ket{n-k}\bra{m-k} -
          \frac{1}{2} \ket{n-m}\bra{0}, &
          n>m.
        \end{array}
      \right.
      \label{DissipatorSplusInversion}
    \end{align}
    To invert $\mathrm{diag}(X)$, we introduce diagonal matrices
    $\phi(k)_{i j}=\delta_{ij} \theta(k-i)$, $i,j,k=1,2,\dots, 2 s$,
    where $\theta(n)$ is a discrete Heaviside step function. Then,
    \begin{align} \mathcal{D}_{L}^{-1}\mathrm{diag}(X) =&
      \sum_{k=1}^{2s} \tr \left(\phi(k) \mathrm{diag}(X) \right)
      \left( |k-1\rangle \langle k-1 | - \ketbra{0}{0}\right).
    \end{align}
  \end{widetext}
  We conclude that $\tr X=0$ is a sufficient condition to construct
  $\mathcal{D}_{L}^{-1}X$.
\end{proof}

\begin{proof}[Proof of theorem \ref{theorem1}]
  We start by assuming that there exists an expansion of the NESS in
  powers of $1/\Gamma$,
  \begin{align}
    \rho_{\rm{NESS}}(\Gamma) = \sum_{m=0}^\infty
    \frac{\rho^{(m)}}{\Gamma^m},
    \label{NessExpansion}
  \end{align}
  convergent for sufficiently large $\Gamma$, i.e.,
  $1/\Gamma<1/\Gamma_\mathrm{cr}$. We will always suppose to be inside
  the convergence disk.

  According to the hypothesis and the choice made for the basis in
  $\mathcal{H}_0$, we have as $\rho^{(0)}=\ket{e^0}\bra{e^0}\otimes
  \ket{\psi_{\rm{target}}}\bra{\psi_{\rm{target}}}$ with $\ket{e^0}
  \equiv \ket{\psi_\mathrm{Zeno}}$. Substituting the expansion
  (\ref{NessExpansion}) into the Lindblad master equation, and
  comparing order by order, we obtain: $\mathcal{D}_{L} \rho^{(0)}=0$
  and
  \begin{align}
    \rho^{(m+1)} = i \mathcal{D}_L^{-1}[H,\rho^{(m)}]+ M^{(m+1)},
    \label{RecurrenceRho}
  \end{align}
  where $M^{(m+1)} \in \ker {\cal{D}}_L$. Note that due to the
  Hermiticity of $\rho_{\rm{NESS}}$, and the property
  $\tr\rho_{\rm{NESS}}=1$, it follows $\left( M^{(m)} \right)^\dagger
  = M^{(m)}$ and $\tr M^{(m)}=0$ for all $m$.  For the expansion
  (\ref{NessExpansion}) to be consistent, the existence of an operator
  inverse in (\ref{RecurrenceRho}) is required at any order $m>0$.
  Lemma \ref{lemma1} allows us to assume ${L}= J_{0,d_0}$ so that,
  according to lemma \ref{lemma2}, the existence of the inverse
  $\mathcal{D}_L^{-1}$ at the $m$-th order is granted by
  \begin{align}
    \traccazero[H,\rho^{(m)}] =0, \qquad m=0,1,\dots,
    \label{Cond1}
  \end{align}
  where $\traccazero$ denotes the partial trace in Hilbert space
  $\mathcal{H}_0$.  In addition, we require a nontriviality condition,
  which guarantees that the NESS becomes a pure state \textit{only} in
  the Zeno limit, whereas it is a mixed state for any finite
  $\Gamma$. Such a condition can be written as
  \begin{align} [H,\rho^{(0)}] \neq 0,
    \label{CondNontriviality}
  \end{align}
  and amounts to have $\kappa\neq 0$ in
  Eq.~(\ref{ConditionTargetPure}).

  The existence of the expansion~(\ref{NessExpansion}) implies that
  the consistency conditions~(\ref{Cond1}) are to be satisfied at any
  order.  Using the decomposition $H =
  \sum_{j=0}^{d_0-1}\sum_{k=0}^{d_0-1} H_{jk}$, with $H_{jk} =
  |e^{j}\rangle \langle e^k | \otimes h_{jk}$, we have
  \begin{align} [H,\rho^{(0)}] &= [H,\ket{e^0}\bra{e^0} \otimes
    \ket{\psi_{\rm{target}}}\bra{\psi_{\rm{target}}}] \nonumber \\ &=
    \sum_{k=1}^{d_0-1}\left( H_{k0} \rho^{(0)} -\rho^{(0)} H_{0k}
    \right),
    \label{defQ1}
  \end{align}
  so that the condition (\ref{Cond1}) for $m=0$ gives
  \begin{align}
    \traccazero [H,\rho^{(0)}]=
    [h_{00},\ket{\psi_{\rm{target}}}\bra{\psi_{\rm{target}}} ]=0,
  \end{align}
  entailing
  \begin{align}
    h_{00}\ket{\psi_{\rm{target}}} = \lambda \ket{\psi_{\rm{target}}}.
    \label{Conditionh00}
  \end{align}
  Thus $\ket{\psi_{\rm{target}}}$ must be an eigenstate of $h_{00}$,
  which we choose to identify as the eigenstate $\alpha=0$, namely,
  $\ket{\psi_{\rm{target}}} \equiv \ket{0}$ and $\lambda \equiv
  \lambda_0$.

  At the next order of the $1/\Gamma$ power expansion, $\rho^{(1)}$ is
  constructed applying $\mathcal{D}_L^{-1}$ to the commutator
  (\ref{defQ1}). To proceed, we note that for any $k>0$ and for any
  $\alpha,\beta$ we have $\mathcal{D}_L \eket{0}{\alpha}
  \ebra{k}{\beta} =-\frac{1}{2} \eket{0}{\alpha} \ebra{k}{\beta}$, and
  $\mathcal{D}_L \eket{k}{\alpha} \ebra{0}{\beta} =-\frac{1}{2}
  \eket{k}{\alpha} \ebra{0}{\beta}$, where we introduced the notation
  $\eket{j}{\alpha}=\ket{e^j}\otimes\ket{\alpha}$.  We obtain
  \begin{align}
    \rho^{(1)}= M^{(1)} -2 i \sum_{k=1}^{d_0-1}
    \left(H_{k0}\rho^{(0)}-\rho^{(0)}H_{0k} \right).
    \label{def_ro1}
  \end{align}
  Because $M^{(1)}$ is an element of the kernel of $\mathcal{D}_L$, we
  rewrite it as $M^{(1)}= \sum_{\alpha,\beta} M^{(1)}_{\alpha,\beta}
  \eket{0}{\alpha} \ebra{0}{\beta}$, with $M^{(1)}_{\alpha,\beta}$
  unknown coefficients. By plugging expression~(\ref{def_ro1}) into
  Eq.~(\ref{Cond1}) for $m=1$, we get
  \begin{align}
    \traccazero[H,i \rho^{(1)}] &= \traccazero\left( -4
      \sum_{k=1}^{d_0-1} \mathcal{D}_{H_{k0}} \rho^{(0)} +
      [H_{00},M^{(1)}] \right) \nonumber \\ &= 0.
    \label{ResCond1}
  \end{align}
  In the above expression,
  $\mathcal{D}_{H_{k0}}$ is a dissipator of the form
  (\ref{DefDissipator}) with $L\to H_{k0}$, and we used $
  h_{jk}^\dagger=h_{kj}$.  Equation~(\ref{ResCond1}) provides
  ${d_1}^2$ scalar equations $\bra{\alpha} \traccazero[H,i \rho^{(1)}
  ] \ket{\beta}=0$, $\alpha,\beta=0,1,\dots,d_1-1$. In particular, for
  $\alpha=\beta=0$ we obtain, after some algebra,
  \begin{align}
    \bra{0} \traccazero[H,i \rho^{(1)}] \ket{0} =
    \sum_{k=1}^{d_0-1}\sum_{\gamma=1}^{d_1-1}
    |\bra{0}h_{0k}\ket{\gamma}|^2=0,
    \label{ResCond1_00}
  \end{align}
  leading to $\bra{\gamma}h_{j0}\ket{0} =\kappa_j \delta_{\gamma,0}$
  for any $j>0$.  Without loss of generality, we can choose the
  vectors $\ket{e^j}$, with $j>0$, of the orthonormal basis in
  $\mathcal{H}_0$ in such a way that $\kappa_1 \equiv \kappa \neq 0$
  and $\kappa_j =0$ for $j>1$. In this basis, we have
  \begin{align}
    \bra{\gamma}h_{k0}\ket{0} = \lambda \delta_{k,0}
    \delta_{\gamma,0}+ \kappa\delta_{k,1} \delta_{\gamma,0},
    \label{ConditionAllCzero}
  \end{align}
  equivalent to Eq.~(\ref{ConditionTargetPure}) with $\ket{e^1} \equiv
  \ket{\psi_{\rm{Zeno}}^\perp}$.
\end{proof}

\begin{proof}[Proof of theorem \ref{theorem2}]
  According to the hypothesis, and with the identifications $\ket{e^0}
  \equiv \ket{\psi_{\rm{Zeno}}}$, $\ket{e^1} \equiv
  \ket{\psi_{\rm{Zeno}}^\perp}$, $\ket{0} \equiv
  \ket{\psi_{\rm{target}}}$ and $\lambda_0 \equiv \lambda$, the
  Hamiltonian $H$ is such that $h_{00} \ket{0} = \lambda_0 \ket{0}$
  and $h_{10} \ket{0} = \kappa \ket{0}$, $\kappa \neq 0$.

  If the consistency conditions (\ref{Cond1}) were satisfied at any
  order $m$, the $1/\Gamma$ power expansion (\ref{NessExpansion})
  would allow us to state that $\lim_{\Gamma\to\infty}
  \rho_{\rm{NESS}}(\Gamma) = \rho^{(0)}$.  Supposing for a moment that
  this is the case, and observing that $\rho^{(0)}= \ket{e^0}\bra{e^0}
  \otimes \ket{0}\bra{0}$ satisfies Eq.~(\ref{Conditionh00}), in
  virtue of the unicity of the NESS, we achieve the thesis
  $\lim_{\Gamma\to\infty} \rho_{\rm{NESS}}(\Gamma) =
  \ket{\Psi}\bra{\Psi}$, with $\ket{\Psi } = \ket{\psi_{\rm{Zeno}}}
  \otimes \ket{\psi_{\rm{target}}}$.  Therefore, it remains to
  demonstrate that the consistency conditions (\ref{Cond1}) are
  satisfied for any $m>0$.

  Condition (\ref{Cond1}) for $m=1$ is given by Eq.~(\ref{ResCond1})
  and is equivalent to ${d_1}^2$ scalar equations obtained by taking
  the expectation of (\ref{ResCond1}) between any two states
  $\bra{\alpha}$ and $\ket{\beta}$.  The scalar equation obtained for
  $\alpha=\beta=0$, namely, Eq.~(\ref{ResCond1_00}), is now satisfied
  by hypothesis (\ref{ConditionTargetPure}).  We split the remaining
  ${d_1}^2-1$ equations in in three parts:
  \begin{align}
    &i [h_{00},m^{(1)}\Phi]= -2 h_{01} h_{10} \Phi
    + 2 |\kappa|^2 \Phi, \label{C1-1} \\
    &i [h_{00},\Phi m^{(1)}]= -2 \Phi h_{01} h_{10}
    + 2 |\kappa|^2 \Phi, \label{C1-2} \\
    & [h_{00},m^{(1)}_V]=0, \label{C1-3}
  \end{align}
  where we have introduced
  \begin{align}
    & m^{(1)}=\traccazero M^{(1)}=
    \sum_{\alpha=0}^{d_1-1}\sum_{\beta=0}^{d_1-1} M^{(1)}_{\alpha
      \beta} \ket{\alpha}\bra{\beta}, \\
    &m^{(1)}_V=\sum_{\alpha=1}^{d_1-1}\sum_{\beta=1}^{d_1-1}
    M^{(1)}_{\alpha \beta} \ket{\alpha}\bra{\beta}, \\
    &\Phi= \traccazero \rho^{(0)}=\ket{0}\bra{0}.
  \end{align}

  Equation~(\ref{C1-1}) is a closed set of linear equations for the
  $d_1-1$ unknowns $M^{(1)}_{\alpha 0}$, $\alpha>0$, which, using
  (\ref{ConditionAllCzero}), can be written as
  \begin{align}
    \sum_{\ga=1}^{d_1-1} \left( \bra{\alpha}h_{00}\ket{\gamma} -
      \lambda_0 \delta_{\alpha,\gamma} \right) M^{(1)}_{\ga 0} = 2 i
    \kappa \bra{\alpha}h_{01}\ket{0}.
    \label{LinearSystemForM1Phi}
  \end{align}
  The unknowns $M^{(1)}_{\alpha 0}$ are uniquely determined from
  (\ref{LinearSystemForM1Phi}) if $\det{\parallel
    \bra{\alpha}h_{00}\ket{\gamma} - \lambda_0
    \delta_{\alpha,\gamma} \parallel} \neq 0$, or, equivalently, if
  the eigenvalue $\lambda_0$ of the $d_1\times d_1$ matrix $h_{00}$ is
  nondegenerate. In the basis where $h_{00}$ is diagonal,
  $\bra{\alpha}h_{00}\ket{\gamma} = \delta_{\alpha,\gamma}
  \lambda_\alpha$, so we obtain $M^{(1)}_{\alpha 0}= 2 i \ka \bra{\al}
  h_{01} \ket{0}/(\la_\al-\la_0)$, or, in matrix form,
  \begin{align}
    m^{(1)} \Phi - M^{(1)}_{0 0} \Phi = 2 i \Lambda h_{01} h_{10}
    \Phi,
    \label{ResM1PhiWithLambda}
  \end{align}
  where $\Lambda$ is given by Eq.~(\ref{DefLambda}).

  Equation~(\ref{C1-2}) is the Hermitian conjugate of Eq.~(\ref{C1-1})
  and leads to $M^{(1)}_{0 \alpha }= \overline{M^{(1)}_{\alpha 0}}$.

  Equation~(\ref{C1-3}) is a set of $(d_1-1)^2$ equations that do not
  determine $m^{(1)}_V$ completely but restrict it to have common
  eigenvectors with $h_{00}$. If the spectrum of $h_{00}$ is
  nondegenerate, then $m^{(1)}_V$ is automatically diagonal in the
  basis $\ket{\al}$; otherwise, it can be made diagonal,
  \begin{align}
    &m^{(1)}_V= \sum_{\alpha=1}^{d_1-1} \mu_\alpha \ketbra{\al}{\al},
    \label{Eigenbasis_m1}\\
    &h_{00}=\sum_{\alpha=0}^{d_1-1} \la_\alpha \ketbra{\al}{\al},
  \end{align}
  where all $\lambda_\alpha,\mu_\alpha$ are real.  To determine
  $\mu_\alpha$, additional relations are needed which come from
  Eq.~(\ref{Cond1}) for $m=2$. Calculating $\rho^{(2)}=i
  \mathcal{D}_L^{-1}[H,\rho^{(1)}]+M^{(2)}$, with $M^{(2)}=
  \sum_{\alpha \beta} M^{(2)}_{\alpha \beta} \eket{0}{\alpha}
  \ebra{0}{\beta}$, we obtain
  \begin{align}
    \rho^{(2)} =\ & M^{(2)} - 2 i [H,M^{(1)}]+ 8 \mathcal{D}_H
    \rho^{(0)} \nonumber \\ & + 4 \left( \ka^2\rho^{(0)} - H_{10}\Om
      H_{01} \right)\label{Res_ro2}.
  \end{align}
  We then split the set of ${d_1}^2$ equations (\ref{Cond1}) for $m=2$
  into five parts:
  \begin{align}
    0 &= \traccazero [H, \rho^{(2)}] \nonumber \\ &=
    \sum_{\alpha=0}^{d_1-1}\sum_{\beta=0}^{d_1-1} Z_{\al\be}
    \ketbra{\al}{\be} \nonumber \\ &= Z_{0 0} \ketbra{0}{0}
    +\sum_{\al=1}^{d_1-1} Z_{\al 0} \ketbra{\al}{0}
    +\sum_{\al=1}^{d_1-1} Z_{0 \al} \ketbra{0}{\al} \nonumber \\
    &\quad + \sum_{\al=1}^{d_1-1} \sum_{\be=1}^{d_1-1}
    (1-\delta_{\al,\be}) Z_{\al\be} \ketbra{\al}{\be}
    +\sum_{\al=1}^{d_1-1} Z_{\al\al} \ketbra{\al}{\al}.
    \label{Zsplitting}
  \end{align}
  First, consider the $d_1-1$ equations $Z_{\al\al}=0$. After some
  algebra, using (\ref{Eigenbasis_m1}) we obtain
  \begin{align}
    & Z_{\al\al}= 4 i \bra{\al} \sum_{k=1}^{d_0-1}
    \mathcal{D}_{h_{k0}} m^{(1)}_V \ket{\al} = \sum_{\be=1}^{d_1-1}
    K_{\al \be} \mu_\be=0.
    \label{Eq Cond2diag}
  \end{align}
  where $K_{\al \be}$ are given by~(\ref{DefK}). Now, if $\det K\neq
  0$, then all $\mu_\be=0$, and therefore $m^{(1)}_V=0$.

  Consider now $Z_{00}=0$. Using $Z_{00}\Phi = \Phi \traccazero
  [H,\rho^{(2)}] \Phi$, after some algebra we obtain
  \begin{align}
    Z_{00}\Phi &= - 2i \Phi \traccazero [H,[H,M^{(1)}]] \Phi \nonumber
    \\ &= 2 i \left( \Phi m^{(1)} h_{01}h_{10} \Phi + \Phi
      h_{01}h_{10} m^{(1)} \Phi \right).
    \label{Eq Z00}
  \end{align}
  Substituting $h_{01}h_{10} \Phi$ and $\Phi h_{01}h_{10}$ from
  Eqs.~(\ref{C1-1}) and (\ref{C1-2}), we find $Z_{00}\Phi=4 i |\ka|^2
  \Phi m^{(1)} \Phi =0$, so that
  \begin{align}
    M^{(1)}_{00} =0.
    \label{Eq M1_00=0}
  \end{align}
  The matrix $M^{(1)}$ is thus completely determined
  \begin{align}
    M^{(1)} &= \ketbra{e^0}{e^0}\otimes \left( m^{(1)} \Phi + \Phi
      m^{(1)} \right) \nonumber \\ &= 2 i \ketbra{e^0}{e^0}\otimes
    \left( \Lambda h_{01}h_{10} \Phi - \Phi h_{01}h_{10} \Lambda
    \right).
    \label{ResultM1}
  \end{align}

  The set of equations $Z_{\al\be}=0$, for $\al,\be>0$ with $\alpha
  \neq \beta$, after some algebra is brought in the form
  \begin{align}
    \bra{\al} [h_{00},{\tilde{m}}^{(2)}_V] \ket{\be}=0,
  \end{align}
  where
  \begin{align}
    \tilde{m}^{(2)}_V= m^{(2)}_V- m^{(1)}\Phi m^{(1)},
    \label{DefM2tilde}
  \end{align}
  \begin{align}
    & m^{(2)}=\traccazero M^{(2)}=
    \sum_{\alpha=0}^{d_1-1}\sum_{\beta=0}^{d_1-1} M^{(2)}_{\alpha
      \beta} \ket{\alpha}\bra{\beta}, \\
    &m^{(2)}_V=\sum_{\alpha=1}^{d_1-1}\sum_{\beta=1}^{d_1-1}
    M^{(2)}_{\alpha \beta} \ket{\alpha}\bra{\beta},
  \end{align}
  and it entails the diagonal nature of the matrix $\tilde{m}^{(2)}_V$
  in the eigenbasis of $h_{00}$,
  \begin{align}
    {\tilde{m}}^{(2)}_V = \sum_{\alpha=1}^{d_1-1} q_\alpha
    \ketbra{\al}{\al}.
    \label{Eigenbasis_m2tilde}
  \end{align}

  The equations $Z_{\al 0}=0$ and $Z_{0 \al}=0$, $\alpha >0$,
  determine the first column and the first row of the matrix
  $M^{(2)}$, except for the element $M^{(2)}_{00}$.  Using
  Eqs. (\ref{ConditionAllCzero}), (\ref{C1-1}) and (\ref{C1-2}), these
  equations can be written in the form
  \begin{align}
    [h_{00},m^{(2)}- 2 \lambda_0 m^{(1)}]\Phi = Q \Phi- \Phi Q
    \Phi, \label{Eq Cond2adiag}
  \end{align}
  where
  \begin{align}
    Q= - 4 i \sum_{k=1}^{d_0-1} \mathcal{D}_{h_{k0}} m^{(1)} \Phi - 8
    |\ka|^2 h_{11} \Phi + 4 \ka \sum_{k=1}^{d_0-1} h_{0k} h_{k1} \Phi,
    \nonumber
  \end{align}
  and the Hermitian conjugate equation for $\Phi m^{(2)}$. The linear
  problem (\ref{Eq Cond2adiag}) is readily solvable as
  \begin{align}
    \left( m^{(2)}- 2 \lambda_0 m^{(1)}\right) \Phi - M^{(2)}_{00}
    \Phi = \Lambda Q \Phi. \label{Res Cond2adiag}
  \end{align}

  To completely determine the matrix $M^{(2)}=\ketbra{e^0}{e^0}\otimes
  m^{(2)}$, with
  \begin{align}
    &m^{(2)}=\Phi m^{(2)}+ m^{(2)} \Phi + m^{(2)}_V- M^{(2)}_{00}
    \Phi, \label{Def M2}
  \end{align}
  we need the diagonal entries of $m^{(2)}_V$, which, according to
  Eq.~(\ref{Eigenbasis_m2tilde}), are given by $ M^{(2)}_{\al \al}=
  q_\alpha + M^{(1)}_{\al 0} M^{(1)}_{0 \al}$, $\al>0$.  To find $
  q_\alpha$ it is necessary to go to the third order of the expansion
  and write down the set of equations $\bra{\al} \traccazero [ H,
  \rho^{(3)} ] \ket{\al}=0$.  After quite tedious but straightforward
  calculations, using Eq.~(\ref{Eq Cond2adiag}), we obtain a closed
  set of equations for $ q_\be$ of the form
  \begin{align}
    \sum_{\be=0}^{d_1-1} K_{\al \be} q_\be + \bra{\al} \mathcal{F
    }\Phi \mathcal{F }^\dagger \ket{\al} =0, \label{Eq Cond3diag}
  \end{align}
  where
  \begin{align}
    \mathcal{F } = \sum_{k=1}^{d_0-1}\left( 2 i h_{k1} h_{10} + [
      m^{(1)},h_{k0} ] \right) \nonumber
  \end{align}
  and $K_{\al \be}$ is given by Eq.~(\ref{DefK}). The solution of the
  linear problem (\ref{Eq Cond3diag}) exists if $\det K \neq
  0$. Finally, $M^{(2)}_{00}$ is found from the requirement $\tr
  M^{(2)}=0$. At this point, the first two orders $\rho^{(1)}$ and
  $\rho^{(2)}$ of the NESS expansion~(\ref{NessExpansion}) are
  completely determined from Eqs.~(\ref{def_ro1}), (\ref{Res_ro2}),
  (\ref{ResultM1}) and (\ref{Def M2}).

  We remark that all systems of equations obtained until now, arising
  in the first three orders of the NESS expansion in powers of
  $1/\Gamma$, Eqs.~(\ref{C1-1}), (\ref{C1-2}), (\ref{Eq Cond2adiag})
  and Eqs.~(\ref{Eq Cond2diag}), (\ref{Eq Cond3diag}), are governed by
  two $(d_1-1)\times (d_1-1)$ matrices: the matrix $\tilde h_{00}$
  with elements $(h_{00})_{\al\be} - \lambda_0 \delta_{\alpha,\beta}$,
  and the matrix $K$ given by Eq.~(\ref{DefK}).  The simultaneous
  invertibility of $\tilde h_{00}$ and $K$, namely,
  \begin{align}
    & (\det K) (\det {\tilde h}_{00}) \neq 0, \label{Cond
      Invertibility}
  \end{align}
  guarantees the existence and uniqueness of $\rho^{(1)}$ and
  $\rho^{(2)}$, $\rho^{(0)}$ being fixed by the targeting condition
  (\ref{ConditionTargetPure}).  We checked that no further
  restrictions arise in the higher orders of the $1/\Gamma$ expansion.
  Thus, the condition (\ref{Cond Invertibility}), together with the
  main condition (\ref{ConditionTargetPure}) form a complete set of
  sufficient conditions (apart from NESS uniqueness) for the NESS to
  approach the pure target state in the Zeno limit.

  Finally, to estimate the convergence of $\rho_{\rm{NESS}}(\Gamma)$
  to the pure target state, we investigate the purity $1-
  \tr\rho^2_\mathrm{NESS}(\Gamma) = {\Gamma_\mathrm{ch}}^2/\Gamma^2+
  O(\Gamma^{-3})$.  A somewhat lengthy but straightforward calculation
  yields
  \begin{align}
    {\Gamma_\mathrm{ch}}^2 = -2 \sum_{\alpha=1}^{d_0-1} \left|
      M^{(1)}_{0\alpha} \right|^2 - 2 M^{(2)}_{00}.
    \label{PurenessCharacteristicGamma}
  \end{align}
  From Eq.~(\ref{Eigenbasis_m2tilde}), using $M^{(2)}_{00}=-
  \sum_{\al>0} M^{(2)}_{\al\al}$, we readily find
  ${\Gamma_\mathrm{ch}}^2 = \sum_{\al>0} (M^{(2)}_V)_{\al \al} =
  \sum_{\al>0} q_\al$, where $q_\al$ are determined by Eq.~(\ref{Eq
    Cond3diag}). Multiplying (\ref{Eq Cond3diag}) by $K^{-1}$ and
  using (\ref{ResM1PhiWithLambda}), we obtain
  Eq.~(\ref{ResGammaCharacteristic}).
\end{proof}


Some remarks are in order about the invertibility of the matrix $K$.
According to Eq.~(\ref{DefK}), $K$ is a real matrix. It has
nonnegative off-diagonal elements and its diagonal elements satisfy
$K_{\al\al}+ \sum_{\be \neq \al} K_{\al\be} = -d_\al$, where
$d_\al=\sum_{k>0} |\bra{0} h_{k0}\ket{\al}|^2$.  By the Gershgorin
circle theorem, every eigenvalue of $K$ lies within at least one of
the disks on the complex plane centered at $K_{\al\al}$ and with
radius $\sum_{\al \neq \be} K_{\al\be}$. All the disks lie completely
in the negative half-plane, and some of them (those for which $d_\al
=0$) touch the origin.  If $d_\al >0$ for all $\al>0$, all the
eigenvalues of $K$ are strictly negative and $K$ is invertible. If
$d_\al=0$ for any $\al>0$, then $K$ becomes a stochastic matrix and
has an eigenvalue $0$, i.e., it is not invertible. If only some
$d_\al=0$, then, by the Gershgorin theorem, an eigenvalue $0$ is not
excluded. Calculating $\det K$ via the Leibnitz formula, we find that
for $\det K=0$, all Leibnitz terms must vanish separately. This
allows, in principle, for a complete classification of the non
invertibility points of $K$.


%

\end{document}